\documentclass[journal]{IEEEtran}
\usepackage[utf8]{inputenc}
\usepackage{graphicx}
\usepackage{cite}
\usepackage{balance}
\usepackage{amsmath,amsthm,amssymb,mathrsfs}
\usepackage{epstopdf}
\usepackage{threeparttable}
\usepackage{multirow}
\usepackage{hyperref}
\usepackage{color}
\usepackage{caption}
\usepackage{mathtools}
\usepackage{lscape}
\newcounter{MYtempeqncnt}
\usepackage{subfig}
\usepackage{diagbox}
\newtheorem{theorem}{Theorem}

\usepackage{suffix}
\usepackage{mathtools}
\ifCLASSINFOpdf
\else
\fi

\DeclarePairedDelimiterX\MeijerM[3]{\lparen}{\rparen}%
{\,#3\delimsize\vert\begin{smallmatrix}#1 \\ #2\end{smallmatrix}}

\newcommand\MeijerG[8][]{%
  G^{\,#2,#3}_{#4,#5}\MeijerM[#1]{#6}{#7}{#8}}

\WithSuffix\newcommand\MeijerG*[7]{%
  G^{\,#1,#2}_{#3,#4}\MeijerM*{#5}{#6}{#7}}

\newcommand{\subparagraph}{}
\usepackage{titlesec}

\titlespacing*{\subsubsection}
{0pt}{.086cm}{.1cm}

\begin{document}
\title{On Performance Characterization of Cascaded Multiwire-PLC/MIMO-RF Communication System}
\author{Yun Ai,~\IEEEmembership{Member, IEEE},~Long Kong,~Michael Cheffena,\\~Symeon Chatzinotas,~\IEEEmembership{Senior Member, IEEE}, and Björn Ottersten, \IEEEmembership{Fellow, IEEE}
\thanks{Y. Ai and M. Cheffena are with the Norwegian University of Science and Technology (NTNU), 2815 Gjøvik, Norway (email: \{yun.ai,michael.cheffena\}@ntnu.no).}
\thanks{L. Kong, S. Chatzinotas, and B. Ottersten are with the Interdisciplinary Centre for Security Reliability and Trust (SnT), University of Luxembourg, L-1855 Luxembourg (email: \{long.kong, symeon.chatzinotas, bjorn.ottersten\}@uni.lu).}
\thanks{Manuscript received xxxx xxxx.}}


\maketitle

\begin{abstract}
The flexibility of radio frequency (RF) systems and the omnipresence of power cables potentially make the cascaded power line communication (PLC)/RF system an efficient and cost-effective solution in terms of wide coverage and high-speed transmission. This letter proposes an opportunistic decode-and-forward (DF)-based multi-wire/RF relaying system to exploit the advantages of both techniques. The outage probability, bit error rate, and system channel capacity are correspondingly chosen to analyze the properties of the proposed system, which are derived in closed-form expressions and validated via Monte-Carlo simulations. One can observe that our proposed system outperforms the wireless-only system in terms of coverage and data rate, especially when there exists a non-line-of-sight (NLoS) connection between the transmitter and receiver pair.
\end{abstract}



\IEEEpeerreviewmaketitle

\section{Introduction \label{sec:introduction}}
\IEEEPARstart{F}{ast} and reliable wireless local area networks (WLAN), especially WiFi, have facilitated the way people access the Internet tremendously and become an integral part of people's life since its beginning in the 1990s. For indoor wireless systems, the hotspot coverage can be limited to a single room with concrete or metal walls that block radio waves from propagating outside \cite{sayed2014narrowband}. The demands of user traffic in the era of Internet of Things (IoT) and connected devices raises the challenges in terms of blind spots and slower rates at longer range encountered in today's indoor wireless solutions \cite{lai2010wireless}.

With the omnipresence of power cables inside a building, the most economical and efficient solution might be a hybrid network combining the wireless and power line communication (PLC) systems in either parallel \cite{lai2010wireless, sayed2014narrowband, lai2012using} or cascaded \cite{mckeown2015priority, gheth2018hybrid, mathur2018performance} architectures. The cascaded PLC/RF system has been initially investigated in \cite{mckeown2015priority, gheth2018hybrid, mathur2018performance}. However, the works are limited to a simple dual-hop setup with point-to-point transmission for the PLC link; and the radio frequency (RF) sub-system cannot be readily extended to the multiple-input and multiple-output (MIMO) setup, which has already been widely used and has become the de-facto configuration in practical WLAN systems. For the PLC sub-system, the unconventional nature of the PLC channel makes fast and reliable transmission of data via such channels a crucial challenge \cite{gheth2018hybrid, di2011noise, qian2018performance, mathur2018performance}. Thus, the single-input and single-output (SISO) setup of the PLC communication in \cite{gheth2018hybrid} and \cite{mathur2018performance} are severely limited in data rate due to PLC channel impairments. Thus, the PLC sub-system becomes the bottleneck in improving the performance of the whole system. Moreover, MIMO over PLC has already been incorporated to the multiwire PLC standard G.hn to enhance data rate and signaling distance \cite{oksman2009g}, however which is not considered in previous works \cite{mckeown2015priority, gheth2018hybrid, mathur2018performance}.

Different from the simple point-to-point structures considered in previous works \cite{mckeown2015priority, gheth2018hybrid, mathur2018performance} and also aiming to further improve the performance system of cascaded PLC/RF system, we investigate a practical structure of a cascaded multiwire-PLC/MIMO-RF system with opportunistic relaying in this paper. In this context, the hybrid cascaded PLC/RF system consists of one base node that is connected to the Internet via fiber and a number of PLC/wireless interference relaying nodes that are connected to the base with PLC backhaul and distributed all over the place. Besides the advantage of increased system coverage, the PLC backhaul also spares the wireless channel of the relaying node to make it dedicated to servicing end users instead of sharing resources (time/bandwdith) in the relay-base communication. 

The main contributions of this paper are given as follows:
\begin{enumerate}
\item  A practical integrated PLC/RF architecture with the advantages of low installation cost and improved performance is proposed and thereafter investigated.
\item The statistics of the equivalent end-to-end signal-to-noise ratio (SNR) of the proposed system are derived.
\item Based on the derived statistics, theoretical analysis of the outage probability, bit error rate (BER), and channel capacity are conducted.
\item We analyze the performance of the proposed system under different configurations of the PLC and RF sub-systems, which shows improved performance in terms of data rates and coverage.
\end{enumerate}

\emph{Notations}: $\Phi(\cdot)$ is the cumulative distribution function (CDF) of a zero-mean Gaussian random variable (RV) with unit variance \cite[Eq.~(8.250.1)]{jeffrey2007table}, $\| \cdot \|_{F}^{2}$ is the Frobenius norm, $\Gamma(\cdot, \cdot)$ and $\Gamma(\cdot)$ are the upper incomplete Gamma and Gamma functions \cite[Chpt.~8.3]{jeffrey2007table}, respectively. $G_{p,q}^{m,n}{\cdot}^{\cdot}(\cdot|:)$ is the Meijer $G$-function \cite[Chpt.~9.3]{jeffrey2007table}. 
\section{System and Channel Models \label{sec:system_model}}
\begin{figure}[!t]
\centering{\includegraphics[width=0.85\columnwidth]{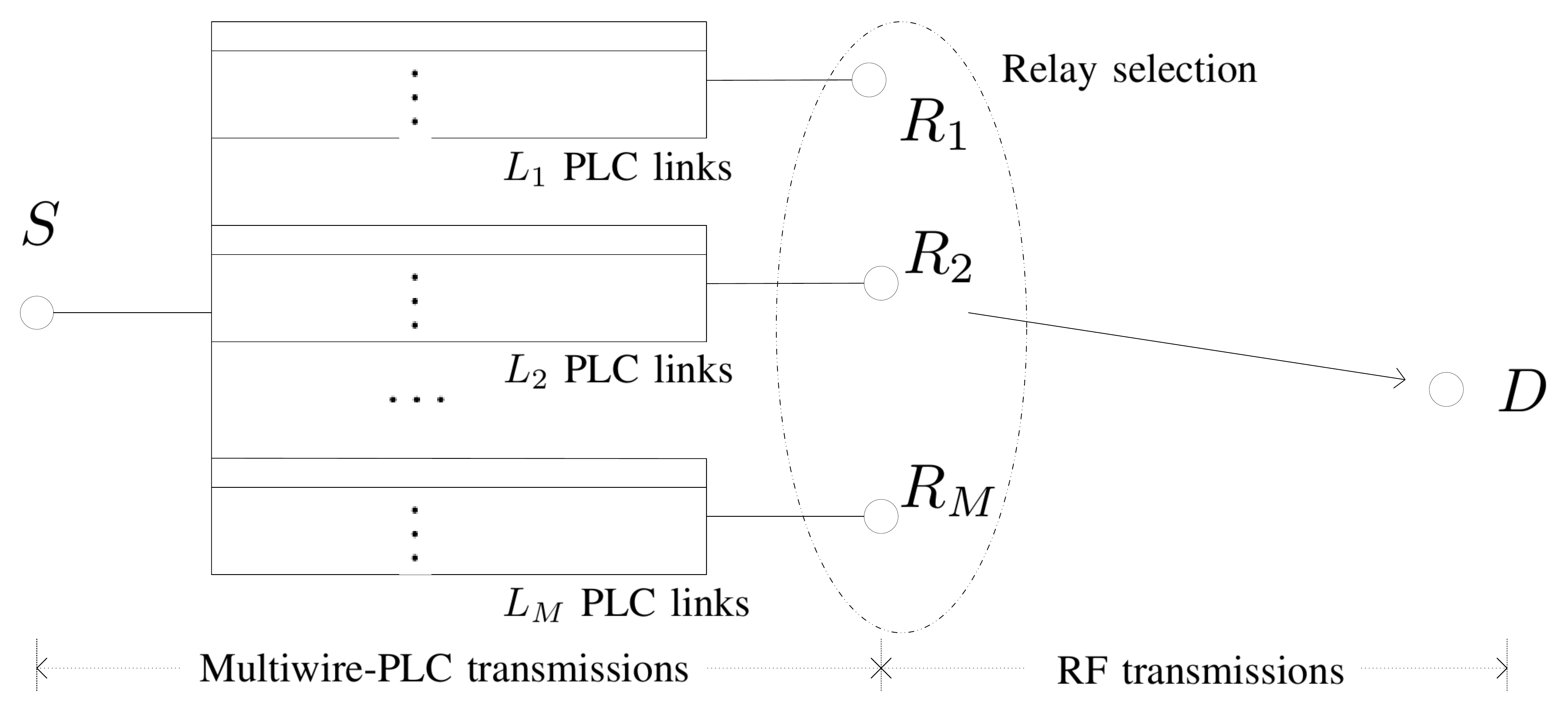}}
\caption{ Illustration of a cascaded multiwire-PLC/MIMO-RF system.}
\label{fig:system}
\end{figure}
We consider the data communication in a big house where the WiFi router (base node) and the users (desired receivers) are located far away and also blocked by concrete walls. The users will experience low data rate due to non-line-of-sight (NLoS) transmission and long transmission distance. With the hybrid system, shown in Fig. \ref{fig:system}, the communication between the base node and the desired receiver located far away operates in two phases. 

In the first phase, the data is sent via PLC links from the base node \emph{S} to the $M$ relaying nodes $R_{m}$, $m = 1, \cdots, M$, and $M > 1$, that are in proximity to the desired receiver \emph{D} with line-of-sight (LoS). To cope with the adverse effects of PLC channels, we assume that in the first phase the PLC sub-system employs a single transmitter and $L_m$-branch receiver diversity at each relaying node $R_{m}$\footnote{For simplicity, we assume that the numbers of branches between the base node and all PLC receivers are the same. However, it is straightforward to extend the analysis to the scenario where the diversities are different.}. In practice, physical power cables, i.e., the neutral, live, and ground wires as well as different paths can be utilized to obtain the multiple diversity branches for simultaneous data transmission. Furthermore, the opportunistic relay selection scheme is applied among the decode-and-forward (DF) relay nodes, i.e., the node having the best link quality will be chosen to decode and then forward the data. Subsequently in the second stage, the selected relay node communicates with the desired receiver using orthogonal space-time block coding (OSTBC) over MIMO wireless channels to improve performance through spatial multiplexing and diversity.
\subsection{Diversity Transmission with Relay Selection over Multi-Wire PLC Links \label{subsec:plc_transmission}}
The received signal at node $R_{m}$, $m= 1, \cdots, M$, via the $l$-th branch of the PLC channel can be written as
\begin{align}
y_{R_{m,l}} =  \sqrt{\mathcal{P}_{0}  \mathcal{L}_{0}} h_{R_{m,l}} s + w_{R_{m,l}},
\label{eq:received_signal}
\end{align}
where $\mathcal{P}_{0}$ is the transmit power, $\mathcal{L}_{0}$ is channel attenuation given as $\exp\left(-2(\alpha_{1} + \alpha_{2}f^{k})d\right)$ with constants $\alpha_{1}$ and $\alpha_{2}$, frequency $f$, and distance $d$ \cite{ai2016capacity}, $s$ is the transmit signal with unit energy, $w_{R_{m,l}}$ is the channel noise, and $h_{R_{m,l}}$, $l = 1, \cdots, L$, are independent and identically distributed (i.i.d.) log-normal channel gains for $l$-th PLC branch between \emph{S} and $R_{m, l}$ with probability density function (PDF) given as
\begin{align}
f_{h_{R_{m,l}}}(x) = \frac{ 1 }{\sqrt{ 2 \pi \sigma^2} x }  \exp\!\left(  - \frac{ ( \ln x - \mu  )^2 }{ 2 \sigma^2 } \right),
\label{eq:plc_pdf}
\end{align}
where $\mu$ and $\sigma^2$ are the average and variance of the Gaussian random RV $\ln x$, respectively. According to \cite{gheth2018hybrid, di2011noise, qian2018performance, mathur2018performance}, the PLC channel noise is a superposition of the background additive white Gaussian noise (AWGN) and Gaussian-distributed impulsive component. Therefore, the PLC noise is a Gaussian RV with the following variance:
$\sigma_{0}^{2} = ( 1 - p  ) \cdot \sigma_{g}^{2} + p \cdot ( \sigma_{g}^{2} + \sigma_{i}^{2} )$,
where $\sigma_{g}^{2}$ and $\sigma_{i}^{2}$ are the powers of the background AWGN and impulsive noise, respectively; and $p$ is the arrival probability of the impulsive noise component.

After receiving the signals from the $L$ branches, the relaying node applys the maximal ratio combining (MRC) to exploit diversity. Then, the received SNR at node $R_{m}$ is
\begin{align}
\gamma_{R_{m}} = \sum_{l=1}^{L} \frac{ \mathcal{P}_{0}  \mathcal{L}_{0}  }{ \sigma_{0}^{2} } h_{R_{m,l}}^{2} =  \sum_{l=1}^{L} \overline{\gamma}_{0}  h_{R_{m,l}}^{2} ,
\label{eq:mrc_snr}
\end{align}
where $\overline{\gamma}_0 = \frac{ \mathcal{P}_{0}  \mathcal{L}_{0}  }{ \sigma_{0}^{2} }$. Using the results in \cite{beaulieu2004highly}, $\gamma_{R_{m}}$ can be accurately approximated as the log-normal sum distributed RV with PDF and CDF
\begin{align}
f_{ \gamma_{R_{m}} }( x ) & =  \frac{ a_{1} a_{2}  (  \frac{x}{\overline{\gamma}_{0} } )^{ -( \frac{a_{2}}{\lambda} + 1 ) } }{ \lambda \sqrt{2\pi} \overline{\gamma}_{0} }  \exp\!\left( - \frac{ [ a_{0} \! - \! a_{1} (  \frac{x}{\overline{\gamma}_{0} } )^{ - \frac{a_{2}}{\lambda} } ]^2 }{ 2 } \right),  \label{eq:plc_mrc_pdf} \\
F_{ \gamma_{R_{m}} }( x ) & =  \Phi\!\left( a_{0} - a_{1}\Bigl( \frac{x}{\overline{\gamma}_{0}} \Bigr)^{ -\frac{a_{2}}{\lambda} }  \right) ,
\label{eq:plc_mrc_cdf}
\end{align}
where $\lambda =  \frac{ \ln\! 10 }{ 10 }$, the mapping between parameters $\mu$, $\sigma_{0}^{2}$ in (\ref{eq:plc_pdf}), $L$ in (\ref{eq:mrc_snr}) and constants $a_{0}$, $a_{1}$, and $a_{2}$ are detailed in \cite{beaulieu2004highly}.

Next, applying relay selection among the $M$ relaying nodes, the CDF of the equivalent SNR $\gamma_{R}$ at the selected relaying nodes of the PLC sub-system is
\begin{align}
F_{\gamma_{R}} (x) =  \prod_{m=1}^{M} \!\! F_{ \gamma_{R_{m}} }( x ) = \left[ \Phi\!\left( a_{0} - a_{1}\Bigl( \frac{x}{\overline{\gamma}_{0}} \Bigr)^{ -\frac{a_{2}}{\lambda} }  \right)  \right]^{M} .
\label{eq:equ_plc_snr_cdf}
\end{align}

Differentiating the above CDF in (\ref{eq:equ_plc_snr_cdf}) yields the following PDF of the SNR $\gamma_{R}$
\begin{align}
f_{ \gamma_{R} }( x )  =  &  \frac{  M    a_{1} a_{2}  (  \frac{x}{\overline{\gamma}_{0} } )^{ -( \frac{a_{2}}{\lambda} + 1 ) } }{ \lambda \sqrt{2\pi} \overline{\gamma}_{0} }     \exp\!\left( - \frac{ [ a_{0} \! - \! a_{1} (  \frac{x}{\overline{\gamma}_{0} } )^{ - \frac{a_{2}}{\lambda} } ]^2 }{ 2 } \right) \nonumber \\   &  \times  \left[ \Phi\!\left( a_{0} - a_{1}\Bigl( \frac{x}{\overline{\gamma}_{0}} \Bigr)^{ -\frac{a_{2}}{\lambda} }  \right)  \right]^{M-1}  .
\label{eq:equ_plc_snr_pdf}
\end{align}
\subsection{MIMO-OSTBC Transmission over RF Links \label{subsec:plc_transmission}}
At the chosen relaying node with the best link quality to the receiver \emph{D}, it selects $K$ transmit symbols, which are encoded with a $N_{R} \times T$ OSTBC matrix $\mathbf{Q}$ and transmitted over $T$ time slots. Then, the received signal at node \emph{D} can be written as
\begin{align}
\mathbf{Y}_{D} =   \sqrt{ \mathcal{P}_{1}  \mathcal{L}_{1} } \mathbf{H}_{RD} \mathbf{Q} + \mathbf{W}_{D} ,
\label{eq:received_signal_mimo}
\end{align}
where $\mathbf{H}_{RD}$ is the $N_{R} \times N_{D}$ channel gain matrix with entries $h_{ij}$ being i.i.d. Nakagami-$m$ RVs of parameters $m$ and $\Omega$, $\mathcal{P}_{1}$ is the transmit power, $\mathcal{L}_{1}$ describes the path loss given by $\frac{c}{  d^n}$ with transmission distance $d$, path loss exponent $n$, and antenna-related constant $c$ \cite{ai2015radio}, $\mathbf{W}_{D} \in  \mathbb{C}^{N_{R} \times T}$ is noise matrix with variance $\sigma_{1}^{2}$. With the OSTBC being used, the MIMO channels are reduced to the rank\{$\mathbf{H}_{RD}$\} parallel SISO channels due to the orthogonalization property of OSTBC. Then, the $k$-th combined symbol can be written as
\begin{align}
y_{D,k} = \| \mathbf{H}_{RD} \|_{F}^{2}  \sqrt{ \mathcal{P}_{1}  \mathcal{L}_{1} } \cdot s_{k} + w_{k} ,
\label{eq:received_signal_mimo2}
\end{align}
where $w_{k}$ is the AWGN with variance $\sigma_{1}^2 \|\mathbf{H}_{RD}\|_{F}^{2}$. Then, the instantaneous SNR at receiver \emph{D} using MIMO-OSTBC with rate $R_{c}$ is 
\begin{equation}
 \gamma_{D} =  \frac{\mathcal{ P}_{1}  \mathcal{L}_{1} \| \mathbf{H}_{RD} \|_{F}^{2} }{ R_{c} N_{R} \sigma_{1}^{2} } = \sum\limits_{i=1}^{N_{D}} \sum\limits_{j=1}^{N_{R}} \frac{ \mathcal{ P}_{1}  \mathcal{L}_{1} |h_{ij}|^{2} }{ R_{c} N_{R} \sigma_{1}^{2} }.
\end{equation}
 Therefore, it immediately follows that the SNR $\gamma_{D}$ is a Gamma RV with the following PDF and CDF:
\begin{align}
f_{\gamma_{D}}(x) & = \frac{ \alpha^{m N_{R} N_{D} }  x^{ m N_{R} N_{D} - 1 }  }{ \Gamma( m N_{R} N_{D} )  \rho^{ m N_{R} N_{D} } }   \exp\Bigl( - \frac{ \alpha x }{ \rho } \Bigr) , \label{eq:mimo_ostbc_pdf} \\
F_{\gamma_{D}}(x) & = 1 - \frac{ \Gamma\!\bigl( m N_{R} N_{D}, \frac{ \alpha x }{ \rho } \bigr) }{ \Gamma( m N_{R} N_{D} ) } ,
\label{eq:mimo_ostbc_cdf}
\end{align}
where $\alpha = \frac{m}{\Omega}$ and $\rho = \frac{ \mathcal{ P}_{1}  \mathcal{L}_{1} }{ R_{c} N_{R} \sigma_{1}^{2} }$. For the simplicity of following notations, let $ \mathcal{A}=m N_{R} N_{D}$.
\section{System Performance Analysis \label{sec:performance_evaluation}}
In this section, we first derive the novel analytical expressions for statistics of equivalent end-to-end SNR, and then we present novel analytical expressions for the considered metrics.

\subsection{Equivalent End-to-end SNR and Outage Performance \label{subsec:equ_snr}}
With the nodes employing DF relaying, the CDF of the equivalent end-to-end SNR $\gamma_{eq}$ is given by
\begin{align}
F_{ \gamma_{eq} }\!( x ) \!  &= \! \mathrm{Pr}( \min( \gamma_{R}, \gamma_{D} ) \! < \! x ) \! \nonumber \\
&= \!  F_{\gamma_{R}}\!( x ) \!  + \!  F_{\gamma_{D}}\!(x)  [ 1 \! - \!  F_{\gamma_{R}}\!(x)  ] .
\label{eq:equ_snr_cdf}
\end{align}
The PDF of the equivalent SNR $\gamma_{eq}$ follows immediately by differentiating (\ref{eq:equ_snr_cdf}) as
\begin{equation} \label{eq:equ_snr_pdf}
f_{ \gamma_{eq} }( x )  =  f_{\gamma_{R}}(x) \bar{F}_{\gamma_D}(x) +    f_{\gamma_{D}} (x)\bar{F}_{\gamma_{R}}(x),
\end{equation}
where $\bar{F}_{\gamma_{D}}( \cdot )$ and $\bar{F}_{\gamma_{R}}( \cdot )$ are the complementary CDF of $\gamma_{D}$ and $\gamma_R$, respectively.

The connection outage occurs when the instantaneous equivalent SNR plunges below some threshold $\gamma_{th}$. Then, the connection outage probability (COP) can be computed by
\begin{align}
\hspace{-2ex} P_{out} &  =   \mathrm{Pr}( x < \gamma_{th} )   =   F_{ \gamma_{eq} }( \gamma_{th} ) \nonumber \\
 = &   1  +  \left(\left[ \Phi\!\left( a_{0}  - a_{1}  \Bigl( \frac{ \gamma_{th} }{\overline{ \gamma }_{0}} \Bigr)^{ -\frac{a_{2}}{\lambda} }  \right)  \right]^{M}  -1\right)     \frac{ \Gamma\!\bigl( \mathcal{A}, \frac{ \alpha \gamma_{th} }{ \rho } \bigr) }{ \Gamma( \mathcal{A} ) }   .
\label{eq:outage_probability}
\end{align}
\subsection{Average Bit Error Rate (BER) \label{subsec:ber}}
Using the following unified expression \cite{ansari2011new}, the average BER for a range of binary modulation schemes can be derived
\begin{equation} \label{eq:avg_ber1b}
\begin{split}
P_{e} & =   \frac{ q^{p} }{ 2 \Gamma(p) }  \int_{ 0 }^{ \infty }   \frac{\exp( - q x )  }{x^{ 1-p }}   F_{ \gamma_{eq} }( x ) \, d x  
=  \frac{  P_{e1} - P_{e2} + P_{e3}  }{ 2 q^{-p}\Gamma(p) } ,
\end{split}
\end{equation}
where $p$ and $q$ are determined by the specific modulation scheme (e.g., $p = q = 0.5$ for binary frequency shift keying (BFSK), $p=0.5$, $q=1$ for binary phase shift keying (BPSK), $p = q = 1$ for differential phase shift keying (DPSK), and $p=1$, $q=0.5$ for binary noncoherent frequency-shift keying (NCFSK)). In (\ref{eq:avg_ber1b}), $P_{e1} \! = \! \int_{ 0 }^{ \infty } \! x^{ p-1 }   \exp( - q x ) dx$, $P_{e2} \! = \! \int_{ 0 }^{ \infty } \! x^{ p-1 }   \exp( - q x )   \frac{ \Gamma\!\left( \mathcal{A}, \frac{ \alpha x }{ \rho } \right) }{ \Gamma( \mathcal{A} ) }   dx$, and $P_{e3} \! = \! \int_{ 0 }^{ \infty } \! x^{ p-1 }   \exp( - q x )  \frac{ \Gamma\!( \mathcal{A}, \frac{ \alpha x }{ \rho } ) }{ \Gamma( \mathcal{A} ) }  \bigl[ \Phi\!\bigl( a_{0} \! - \! a_{1} \!  ( \frac{ x }{\overline{ \gamma }_{0}}  )^{ -\frac{a_{2}}{\lambda} }  \bigr)  \bigr]^{M} dx$.
\begin{theorem}
The BER performance of our considered system setup is therefore given in (\ref{BER}), shown at the top of next page, where $w_{\iota}$ and $y_{\iota}$, ($\iota = 1,\dots, L$), are the weights and zeros of the $L$-order Hermite-Gauss polynomial \cite{steen1969gaussian}.
\begin{figure*}[!t]
\setcounter{MYtempeqncnt}{\value{equation}}
\setcounter{equation}{16}
\begin{equation} \label{BER}
\begin{split}
P_{e} = \frac{1}{2}  - \frac{ 1 }{2\Gamma(p) \Gamma( \mathcal{A} ) }   \MeijerG*{2}{1}{2}{2}{ 1-p, 1 }{ 0, \mathcal{A} }{ \frac{ \alpha  }{ q  \rho } } + \sum_{\iota=1}^{L}     \frac{   w_{\iota} y_{\iota}^{2p-1} }{ \Gamma(p) \Gamma(\mathcal{A}) }  \! \Gamma\left( \!  \mathcal{A}, \!   \frac{ \alpha y_{\iota}^{2} }{ q \rho }  \!  \right)  \left[ \! \Phi\!\left( a_{0} \! - \! a_{1} \!  \left( \frac{ y_{\iota}^{2} }{q \overline{ \gamma }_{0}}  \right)^{ -\frac{a_{2}}{\lambda} }  \right)  \!  \right]^{ \! M }.
\end{split}
\end{equation}
\hrulefill
\vspace{-0.3cm}
\end{figure*}
\end{theorem}
\begin{proof}
See Appendix \ref{Proof_BER}.
\end{proof}
\subsection{Channel Capacity \label{subsec:capacity}}
For investigated setup, the normalized achievable capacity $\overline{C}$ (nats/s/Hz) in Shannon sense can be expressed as \cite{ai2016capacity}
\begin{align}
\overline{C}  & =  \int_{ 0 }^{ \infty } \ln( 1 + x )    [  f_{\gamma_{D}}(x)\bar{F}_{\gamma_{R}}(x)   +  f_{\gamma_{R}}(x)  \bar{F}_{\gamma_{D}}(x)]dx \nonumber \\  &  =  \overline{C}_{1} + \overline{C}_{2} - \overline{C}_{3} ,
\label{eq:avg_cap1}
\end{align}
where $\overline{C}_{1} = \int_{0}^{\infty} \ln( 1 + x ) f_{\gamma_{D}}(x) \, dx $, $\overline{C}_{2} = \int_{0}^{\infty} \ln( 1 + x ) f_{\gamma_{R}}(x) \bar{F}_{\gamma_{D}}(x) \, dx $, and $\overline{C}_{3} = \int_{0}^{\infty} \ln( 1 + x ) F_{\gamma_{R}}(x)  f_{\gamma_{D}} (x) \, dx $.
\begin{theorem}
The average channel capacity of our considered system setup is given in (\ref{Capacity}), shown at the top of next page, where $\tau$, $\varpi_{i}$ and $z_{i}$ are respectively the order of the Hermite polynomial, weight and $i$-th zero of the polynomial.
\begin{figure*}[!t]
\setcounter{MYtempeqncnt}{\value{equation}}
\setcounter{equation}{18}
\begin{align} \label{Capacity}
\overline{C} =& \frac{ \alpha^{\mathcal{A} }  }{ \Gamma( \mathcal{A} )  \rho^{ \mathcal{A} } }   \MeijerG*{3}{1}{2}{3}{ - \mathcal{A} , 1 - \mathcal{A} }{ 0, -\mathcal{A}, -\mathcal{A} }{\!\! \frac{ \alpha  }{ \rho } }  + \sum \limits_{i =1}^\tau \varpi_{i} \frac{M [\Phi(\sqrt{2}z_i)]^{M-1}}{\sqrt{\pi}\Gamma( \mathcal{A})}  \ln \biggl[ 1 + \bar{\gamma}_0 \Bigl( \frac{a_0 - \sqrt{2}z_i}{a_1}\Bigr)^{-\frac{\lambda}{a_2}} \biggr]  \Gamma\biggl(\mathcal{A}, \bar{\gamma}_0 \Bigl( \frac{a_0 - \sqrt{2}z_i}{a_1}\Bigr)^{-\frac{\lambda}{a_2}}  \biggr)  \nonumber \\
& - \frac{ 2 \alpha^{ \mathcal{A} - 1 } }{ \Gamma( \mathcal{A} ) \rho^{ \mathcal{A} - 1 } }  \sum_{\iota=1}^{L} w_{\iota}  y_{\iota} \Bigl( \frac{ \rho y_{\iota}^{2} }{ \alpha } \Bigr)^{  \mathcal{A} - 1 }   \ln\Bigl( 1 + \frac{ \rho y_{\iota}^{2} }{ \alpha } \Bigr)   \left[ \Phi\left( a_{0} - a_{1} \left( \frac{ \rho y_{\iota}^{2} }{ \alpha \bar{\gamma}_{0} } \right)^{ - \frac{a_{2}}{\lambda} } \right) \right]^M .
\end{align}
\hrulefill
\vspace{-0.3cm}
\end{figure*}
\end{theorem}
\begin{proof}
See Appendix \ref{Proof_Capacity}.
\end{proof}
\section{Numerical Results \label{sec:results}}
\begin{figure}[htbp] 
\centering
\subfloat[$P_{out}$ versus different propagation distance]{\includegraphics[width=0.9\columnwidth]{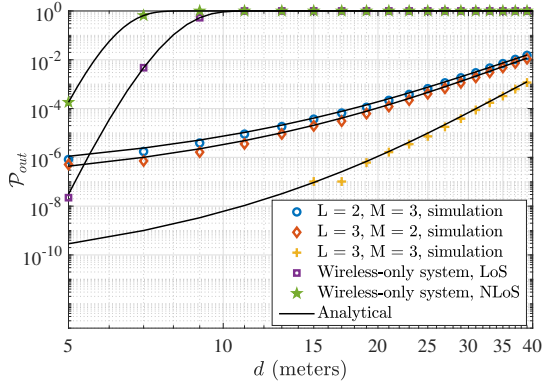}} \\
\subfloat[Average BER versus $L$ and $\rho$]{\includegraphics[width=0.9\columnwidth]{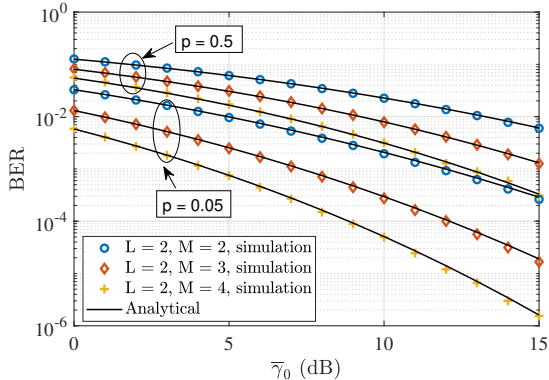}} \\
\subfloat[Channel capacity versus $\bar{\gamma}_0$]{\includegraphics[width=0.9\columnwidth]{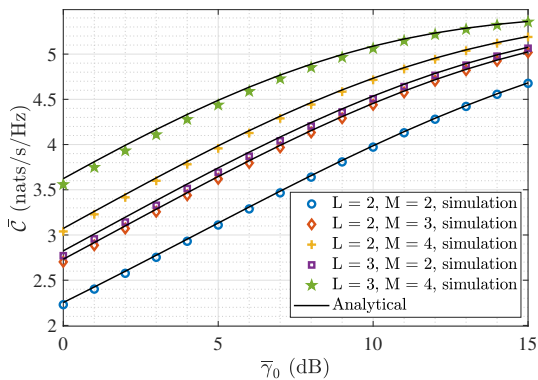}}
\vspace{0.1cm}
\caption{Performance metrics investigations, including outage probability, BER, and channel capacity, of the proposed cascaded multiwire-PLC/MIMO-RF system setup.}
\label{fig}
\vspace{0.2cm}
\end{figure}

%

The numerical results of investigated performance metrics are presented in this section. Unless stated otherwise, the following parameter values are used: $\sigma = 12$ dB, $\mu = 0$ dB, $p = 0.05$ for the PLC part and $R_{c} = 0.85$, $N_{r} = 3$, $N_{d} = 2$, $m = 3$, $\Omega = 1$, $\rho = 20$ dB, $c = 1$ for the RF part.

In Fig. \ref{fig}. (a), the outage probability of the proposed system compared with wireless-only system when the Tx and Rx have line-of-sight (LoS) and NLoS conditions is plotted. We consider an indoor scenario where the path loss exponents are 2.8 and 3.2 for LoS and NLoS scenarios, respectively \cite{ai2015radio}; and the attenuation parameters for PLC channel are $\alpha_{1} = 0.00933$, $\alpha_{2} = 0.0051$, $k = 0.7$ \cite{ai2016capacity}. For simplicity, we assume that for the cascaded setup, the distance of wireless transmission is fixed at 2 meters with varying PLC transmission distance. It is assumed that the PLC operates at 20 MHz \cite{ai2016capacity}. Also, the powers are equally allocated to PLC and wireless parts of the cascaded system with transmit SNR being 25 dB for both links. The optimal power allocation is not presented due to page limit and is subject of future research.

Apparently, one can observe that our analytical results are validated by the Monte-Carlo simulation results, and also from Fig. \ref{fig}. (a) that the cascaded system outperforms the wireless-only system in terms of coverage at the same rate (or in terms of rate at the same coverage), especially in NLoS conditions\footnote{It is worth mentioning that at the time of writing this manuscript, one of the largest telecom equipment vendors in the world launched their latest WiFi solution Q2 that was advertised as the "World's 1st Hybrid Home PLC/Wi-Fi System". The cascaded combination of PLC and WiFi to provide fast connection in large areas (up to 360 $\textrm{m}^{2}$) is one of Q2's selling points.}.

Fig. \ref{fig}. (b) shows the BER performance of BPSK against $\bar{\gamma}_0$ for varying configurations of the investigated setup. It follows from the results that the BER performance improves significantly even as a small number of more subchannels and/or relaying nodes are installed. This performance enhancement of the proposed setup is also reaffirmed in terms of system capacity as illustrated in Fig. \ref{fig}. (c). 

Similar with Figs. \ref{fig}. (a) and (b), the channel capacity of our proposed system setup is presented in Fig. \ref{fig}. (c) with consideration of the impacts of the number of relays $M$ and the number of branches $L$ at the receiver. Obviously, the analytical and simulation results are also found to be consistent in the previous two figures. In addition, the performance gain obtained from the increase of $L$ is fairly larger than simply increasing $M$, which is due to the larger $L$ means better quality of received SNR at relaying nodes.

\section{Conclusion \label{sec:conclusion}}
In this letter, we investigated a multi-wire/MIMO-RF cascaded system under the opportunist DF relaying scheme. Important performance metrics such as outage probability, BER, and capacity were analytically derived in their closed forms, which were also verified using Monte-Carlo simulation results. One can obtain that our proposed system setup is advantageous when enhancing the coverage extension and date rate, especially when Tx and Rx have a NLoS connection.

\appendices
\section{Proof for (\ref{BER})} \label{Proof_BER}
Utilizing \cite[Eq.~(3.381.3)]{jeffrey2007table} for $P_{e1}$ leads to $P_{e1} = \frac{ \Gamma(p) }{ q^{p} }$.

To obtain the solution of $P_{e2}$, we first represent the exponential and Gamma functions terms through their Meijer $G$-functions and then utilize the relation \cite[Eq.~(23)]{adamchik1990algorithm}
\begin{align}
P_{e2} = & \int_{ 0 }^{ \infty } \!   \frac{ x^{ p-1 }    }{ \Gamma( \mathcal{A} ) }  \MeijerG*{1}{0}{0}{1}{-}{ 0 }{ q x }  \MeijerG*{2}{0}{1}{2}{1}{ 0, \mathcal{A} }{ \frac{ \alpha x }{ \rho } } dx \nonumber \\
= & \frac{ 1 }{ q^{p}  \Gamma( \mathcal{A} ) }   \MeijerG*{2}{1}{2}{2}{ 1-p, 1 }{ 0, \mathcal{A} }{ \frac{ \alpha  }{ q  \rho } }  .
\label{eq:avg_ber2}
\end{align}

To solve $P_{e3}$, making the interchange of RVs $qx \! = \! y^{2}$ to yield: 
\begin{align}
P_{e3} \!  =& \! \int_{ 0 }^{ \infty } \exp\left( - y^{2} \right) \! \frac{2y^{2p-1}}{\Gamma( \mathcal{A} ) q^p}     \Gamma\!\left( \mathcal{A}, \frac{ \alpha y^{2} }{ q \rho } \right) \nonumber \\
& \times \left[ \Phi\!\left( a_{0} \! - \! a_{1} \!  \left(\! \frac{ y^{2} }{\overline{  q \gamma }_{0}} \! \right)^{ -\frac{a_{2}}{\lambda} }  \right)  \right]^{M}  dy.
\end{align}
Unfortunately, there exists no closed-form solution to this integral. The integral can be efficiently computed using the modified Gauss-Chebyshev method \cite{steen1969gaussian}, i.e.,
\begin{align}
\int_{0}^{\infty} \exp(-y^{2})  f(y) dy = \sum_{\iota=1}^{L} w_{\iota}  f(y_{\iota}).
\label{eq:gauss_chebyshev}
\end{align}
Hence, $P_{e3}$ can be efficiently evaluated by
\begin{align}
P_{e3} \!\! = \!\!  \sum_{\iota=1}^{L} \frac{ 2 w_{\iota} y_{\iota}^{2p-1} }{ q^{p}\Gamma( \mathcal{A} ) }  \Gamma\left( \!  \mathcal{A}, \!   \frac{ \alpha y_{\iota}^{2} }{ q \rho }  \! \right ) \!  \left[ \! \Phi\!\left( a_{0} \! - \! a_{1} \!  \left( \frac{ y_{\iota}^{2} }{q \overline{ \gamma }_{0}}  \right)^{ -\frac{a_{2}}{\lambda} }  \right)  \!  \right]^{ \! M }.
\label{eq:avg_ber3}
\end{align}
Finally, substituting the expressions of $P_{e1}$, $P_{e2}$, and $P_{e3}$ into (\ref{eq:avg_ber1b}), and making some algebraic manipulations, the proof for (\ref{BER}) is obtained. 
\section{Proof for (\ref{Capacity})}  \label{Proof_Capacity}
Substituting the relevant functions into their Meijer $G$-function counterparts and further capitalizing \cite[Eq.~(23)]{adamchik1990algorithm}, $\overline{C}_{1}$ can be solved as
\begin{align}
\overline{C}_{1} \! = & \! \! \int_{0}^{\infty}  \!\!\! \frac{ \alpha^{\mathcal{A} }  x^{ \mathcal{A} - 1 } }{ \Gamma( \mathcal{A} )  \rho^{ \mathcal{A} } }  \MeijerG*{1}{2}{2}{2}{1,1}{1, 0 }{\! x }  \MeijerG*{1}{0}{0}{1}{-}{ 0 }{\!\! \frac{ \alpha x }{ \rho } } \,dx  \nonumber  \\
= & \frac{ \alpha^{\mathcal{A} }  }{ \Gamma( \mathcal{A} )  \rho^{ \mathcal{A} } }   \MeijerG*{3}{1}{2}{3}{ - \mathcal{A}, 1 - \mathcal{A} }{ 0, -\mathcal{A}, -\mathcal{A} }{\!\! \frac{ \alpha  }{ \rho } } .
\label{eq:avg_cap2}
\end{align}

In the case of $M>1$, we apply the transformation $\sqrt{2}z = a_0 - a_1\bigl(\frac{x}{\bar{\gamma}_0}\bigr)^{-\frac{a_2}{\lambda}}$ and then use the similar approach outlined in \cite[Eqs.~(43-45)]{mathur2018performance}, to obtain the following expression of $\overline{C}_{2}$
\begin{align}
\overline{C}_{2}& =  \int_{0}^{\infty} \ln( 1 + x )  \overline{F}_{\gamma_{D}}(x) \, d F_{\gamma_{R}}(x)  \nonumber  \\
&=  \int_{0}^{\infty} \ln( 1 \! + \! x ) \cdot \frac{ \Gamma\!\bigl( \mathcal{A}, \frac{ \alpha x }{ \rho } \bigr) }{ \Gamma( \mathcal{A} ) } \, d \left[ \! \Phi\!\Bigl( a_{0} \! - \! a_{1}\bigl( \frac{x}{\overline{\gamma}_{0}} \bigr)^{ -\frac{a_{2}}{\lambda} }  \Bigr)  \! \right]^{M} \nonumber  \\
&=   \int_{-\infty}^\infty \exp(-z^2) \cdot g(z)dz \nonumber \\
&\mathop \approx ^{(a)}  \sum \limits_{i =1}^\tau \varpi_{i} \cdot g(z_{i}) ,
\end{align}
where
\begin{equation}
\begin{split}
g(z) = & \frac{M [\Phi(\sqrt{2}z)]^{M-1}}{\sqrt{\pi}\Gamma( \mathcal{A})}  \ln \biggl[ 1 + \bar{\gamma}_0 \Bigl( \frac{a_0 - \sqrt{2}z}{a_1}\Bigr)^{-\frac{\lambda}{a_2}} \biggr] \\
& \times \Gamma\biggl(\mathcal{A}, \bar{\gamma}_0 \Bigl( \frac{a_0 - \sqrt{2}z}{a_1}\Bigr)^{-\frac{\lambda}{a_2}}  \biggr),
\end{split}
\end{equation}
and step $(a)$ is further developed using the modified Gauss-Chebyshev method \cite{steen1969gaussian}. 

Applying the modified Gauss-Chebyshev method on $\overline{C}_{3}$ \cite{steen1969gaussian}, one can similarly evaluate $\overline{C}_{3}$ with the following numerical expression
\begin{align}
\overline{C}_{3}& =   \frac{ 2 \alpha^{ \mathcal{A} - 1 } }{ \Gamma( \mathcal{A} ) \rho^{ m N_{R} N_{D} - 1 } }   \sum_{\iota=1}^{L} w_{\iota}  y_{\iota} \Bigl( \frac{ \rho y_{\iota}^{2} }{ \alpha } \Bigr)^{  \mathcal{A} - 1 } \nonumber \\ & \times \ln\Bigl( 1 + \frac{ \rho y_{\iota}^{2} }{ \alpha } \Bigr)   \left[ \Phi\left( a_{0} - a_{1} \left( \frac{ \rho y_{\iota}^{2} }{ \alpha \overline{\gamma}_{0} } \right)^{ - \frac{a_{2}}{\lambda} } \right) \right]^M .
\label{eq:gauss_chebyshev2}
\end{align}

Finally, substituting the expressions of $\overline{C}_{1}$, $\overline{C}_{2}$, and $\overline{C}_{3}$ into (\ref{eq:avg_cap1}), the proof for the average capacity of the investigated system is accomplished.

\balance
\bibliographystyle{IEEEtran}
\bibliography{cascaded_plc_wireless_system_ref1}

\begin{thebibliography}{10}
\providecommand{\url}[1]{#1}
\csname url@samestyle\endcsname
\providecommand{\newblock}{\relax}
\providecommand{\bibinfo}[2]{#2}
\providecommand{\BIBentrySTDinterwordspacing}{\spaceskip=0pt\relax}
\providecommand{\BIBentryALTinterwordstretchfactor}{4}
\providecommand{\BIBentryALTinterwordspacing}{\spaceskip=\fontdimen2\font plus
\BIBentryALTinterwordstretchfactor\fontdimen3\font minus
  \fontdimen4\font\relax}
\providecommand{\BIBforeignlanguage}[2]{{%
\expandafter\ifx\csname l@#1\endcsname\relax
\typeout{** WARNING: IEEEtran.bst: No hyphenation pattern has been}%
\typeout{** loaded for the language `#1'. Using the pattern for}%
\typeout{** the default language instead.}%
\else
\language=\csname l@#1\endcsname
\fi
#2}}
\providecommand{\BIBdecl}{\relax}
\BIBdecl

\bibitem{sayed2014narrowband}
M.~Sayed and N.~Al-Dhahir, ``Narrowband-{PLC}/wireless diversity for smart grid
  communications,'' in \emph{Proc. IEEE Global Commun. Conf. (GLOBECOM)}.\hskip
  1em plus 0.5em minus 0.4em\relax IEEE, Feb. 2014, pp. 2966--2971.

\bibitem{lai2010wireless}
S.~W. Lai and G.~G. Messier, ``The wireless/power-line diversity channel,'' in
  \emph{Proc. IEEE Int. Conf. Commun. (ICC)}.\hskip 1em plus 0.5em minus
  0.4em\relax Cape Town, South Africa: IEEE, July 2010, pp. 1--5.

\bibitem{lai2012using}
------, ``Using the wireless and {PLC} channels for diversity,'' \emph{IEEE
  Trans. Commun.}, vol.~60, no.~12, pp. 3865--3875, Dec. 2012.

\bibitem{mckeown2015priority}
A.~Mckeown, H.~Rashvand, T.~Wilcox, and P.~Thomas, ``Priority {SDN} controlled
  integrated wireless and powerline wired for smart-home {Internet of
  Things},'' in \emph{Proc. UIC-ATC-ScalCom}.\hskip 1em plus 0.5em minus
  0.4em\relax Beijing, China: IEEE, July 2015, pp. 1825--1830.

\bibitem{gheth2018hybrid}
W.~Gheth, K.~M. Rabie, B.~Adebisi, M.~Ijaz, G.~Harris, and A.~Alfitouri,
  ``Hybrid power-line/wireless communication systems for indoor applications,''
  in \emph{Proc. Int. Symp. Commun. Syst., Netw. Digit. Signal Process.}\hskip
  1em plus 0.5em minus 0.4em\relax Budapest, Hungary: IEEE, July 2018, pp.
  1--6.

\bibitem{mathur2018performance}
A.~Mathur, M.~R. Bhatnagar, Y.~Ai, and M.~Cheffena, ``Performance analysis of a
  dual-hop wireless-power line mixed cooperative system,'' \emph{IEEE Access},
  vol.~6, July 2018.

\bibitem{di2011noise}
L.~Di~Bert, P.~Caldera, D.~Schwingshackl, and A.~M. Tonello, ``On noise
  modeling for power line communications,'' in \emph{Proc. IEEE Int. Symp.
  Power Line Commun. Appl. (ISPLC)}.\hskip 1em plus 0.5em minus 0.4em\relax
  Udine, Italy: IEEE, Apr. 2011, pp. 283--288.

\bibitem{qian2018performance}
Y.~Qian, J.~Li, Y.~Zhang, and D.~N.~K. Jayakody, ``Performance analysis of an
  opportunistic relaying power line communication systems,'' \emph{IEEE Syst.
  J.}, no.~99, pp. 1--4, Dec. 2018.

\bibitem{oksman2009g}
V.~Oksman and S.~Galli, ``{G. hn}: The new {ITU-T} home networking standard,''
  \emph{IEEE Commun. Mag.}, vol.~47, no.~10, pp. 138--145, Oct. 2009.

\bibitem{jeffrey2007table}
I.~S. Gradshteyn and I.~M. Ryzhik, \emph{Table of Integrals, Series, and
  Products}, 8th~ed.\hskip 1em plus 0.5em minus 0.4em\relax Burlington, MA,
  USA: Academic Press, 2015.

\bibitem{ai2016capacity}
Y.~Ai and M.~Cheffena, ``Capacity analysis of {PLC} over rayleigh fading
  channels with colored nakagam-$m$ additive noise,'' in \emph{Proc. IEEE Veh.
  Technol. Conf. (VTC-Fall)}.\hskip 1em plus 0.5em minus 0.4em\relax Montreal,
  Canada: IEEE, Sept. 2016.

\bibitem{beaulieu2004highly}
N.~C. Beaulieu and F.~Rajwani, ``Highly accurate simple closed-form
  approximations to lognormal sum distributions and densities,'' \emph{IEEE
  Commun. Lett.}, vol.~8, no.~12, pp. 709--711, Dec. 2004.

\bibitem{ai2015radio}
Y.~Ai, M.~Cheffena, and Q.~Li, ``Radio frequency measurements and capacity
  analysis for industrial indoor environments,'' in \emph{Proc. European Conf.
  Antennas Propag. (EuCAP)}.\hskip 1em plus 0.5em minus 0.4em\relax Lisboa,
  Portugal: IEEE, April 2015, pp. 1--5.

\bibitem{ansari2011new}
I.~S. Ansari, S.~Al-Ahmadi, F.~Yilmaz, M.-S. Alouini, and H.~Yanikomeroglu, ``A
  new formula for the {BER} of binary modulations with dual-branch selection
  over generalized-${K}$ composite fading channels,'' \emph{IEEE Trans.
  Commun.}, vol.~59, no.~10, Oct. 2011.

\bibitem{steen1969gaussian}
N.~Steen, G.~Byrne, and E.~Gelbard, ``Gaussian quadratures for the integrals
  $\int_{0}^{\infty} e^{-x^{2}}f(x)dx$ and $\int_{0}^{b} e^{-x^{2}}f(x)dx$,''
  \emph{Math. Comput.}, pp. 661--671, May 1969.

\bibitem{adamchik1990algorithm}
V.~Adamchik and O.~Marichev, ``The algorithm for calculating integrals of
  hypergeometric type functions and its realization in {REDUCE} system,'' in
  \emph{Proc. Int. Symp. Symbol. Alg. Comput.}\hskip 1em plus 0.5em minus
  0.4em\relax Tokyo, Japan: ACM, Aug. 1990, pp. 212--224.

\end{thebibliography}

\end{document}